\documentclass[a4paper,11pt,twoside]{article}
\usepackage{amsmath,amsfonts,amssymb}
\usepackage{graphicx}
\usepackage{eucal}
\usepackage{mathrsfs}
\usepackage{theorem}
\usepackage{pifont}
\usepackage {epsfig}
\usepackage {graphicx}
\newcommand{\todo}[1]{{\sffamily To do:}}

\newtheorem{theorem}{Theorem}

\newtheorem {lemma}{Lemma}
\newenvironment{proof}{{\flushleft \emph{Proof}:}}{\ding{110}}

\title{The Immediate Exchange model: an analytical investigation}

\author{Guy Katriel\\ Department of Mathematics, ORT Braude College,\\ Karmiel, Israel\\}

\date{}

\begin{document}

\maketitle

\begin{abstract} We study the Immediate Exchange model, recently introduced by Heinsalu and Patriarca [Eur. Phys. J. B 87: 170 (2014)], who showed by simulations that the wealth distribution in this model converges to a Gamma distribution with shape parameter $2$. Here we justify this conclusion analytically, in the infinite-population limit. An
	  infinite-population version of the model is derived, describing the evolution of the wealth distribution in terms of iterations of a nonlinear operator on the space of probability densities.  It is proved that the Gamma distributions with shape parameter $2$ are fixed points of this operator, and that, starting with an arbitrary 
	  wealth distribution, the process converges to one of these fixed points. We also discuss the mixed model introduced in the same paper, in which exchanges are either bidirectional or unidirectional with fixed probability. We prove that, although, as found by Heinsalu and Patriarca, the equilibrium distribution
	  can be closely fit by Gamma distributions,  the equilibrium distribution for this model is {\it{not}} a Gamma distribution.
\end{abstract}

\section{Introduction}

Kinetic exchange models, which describe exchanges of wealth among a population of agents, have been widely investigated in recent years (see reviews  \cite{ghosh,patriarca1,patriarca2,yakovenko}). 
The central question in these investigations is characterizing the equilibrium wealth distributions that emerge at the global level, given the
microscopic rules of interaction among the agents. 

In a recent paper \cite{heinsalu}, Heinsalu and Patriarca proposed a kinetic exchange model which they called the Immediate Exchange model. In this model pairs of agents
randomly interact, an interaction consisting of each of the agents transferring a random fraction of its wealth to the other agent, where these fractions are 
independent and uniformly distributed in $[0,1]$. Thus, if agents $i,j$ have wealths $x_i,x_j$ prior to the interaction, their wealths following the interaction are
\begin{equation}
\label{ie}x_i'=(1-\epsilon_i)x_i+\epsilon_j x_j,\;\;\; x_j'=(1-\epsilon_j)x_j+\epsilon_i x_i,
\end{equation}
where $\epsilon_i,\epsilon_j$ are independent and  $\epsilon_i,\epsilon_j\sim Uniform([0,1])$.
Based on simulation of this process, Heinsalu and Patriarca have concluded that the wealth distribution converges to a Gamma distribution with shape parameter $2$.
Here we rigorously justify this conclusion, by deriving an infinite population version of the Immediate Exchange model, in which the time-evolution of the wealth 
distribution is described by iteration of a nonlinear operator
on a space of probability distributions (Section \ref{model}), and showing that Gamma distributions with shape parameter $2$ are the fixed points of this operator
(Section \ref{equilibrium}). Furthermore, we prove that, starting from a general wealth distribution, iterations of the operator converge to one of these equilibrium distributions, determined by the mean wealth of the initial distribution, which is conserved (Section \ref{converge}). 

In Section \ref{mixed} we consider another model presented in \cite{heinsalu}, a mixed 
model in which interactions are either bidirectional as in the Immediate Exchange model, or unidirectional as in the Directed Market model \cite{martinez}, each case 
occuring with a certain fixed probability. This model is also studied in the infinite-population limit, and we prove that, despite the fact that the equilibrium distribution can be closely fitted by a Gamma distribution, as shown by numerical simulations in \cite{heinsalu}, the equilibrium distribution is in fact {\it{not}} a Gamma distribution.

\section{Infinite-population version of the Immediate Exchange model}\label{model}

We now formulate the infinite-population discrete-time version of the Immediate Exchange model in the framework of L\'opez, L\'opez-Ruiz and Calbet \cite{lopez1,lopez2}.
The distribution of wealth is described  by a probability density
$p_t(x)$ so that $p_t(x)dx$ is the fraction of the population whose wealth is in the interval $[x,x+dx]$ at time $t=0,1,2,...$.
It is assumed that at each time step (`day') all agents are randomly paired and exchange wealth according to the rule (\ref{ie}). 
Assuming the wealth distribution $p_t(x)$ before the interactions of day $t$ take place is given, we derive the 
probability density $p_{t+1}(x)$ following these interactions, and thus the time-evolution of the distribution of wealth.

The language of probability theory is convenient in deriving the evolution equation.
Let us choose a random agent and let $U$ be a random variable representing this agent's wealth before the interaction on day $t$ takes place, and $X$
its wealth following the interaction. Thus the distributions of $U$ and of $X$ are given by the probability densities $p_t(x)$ and $p_{t+1}(x)$, 
respectively. Let $V$ represent the wealth of the agent with which our focal agent interacted, which is a random variable whose distribution
is also $p_t(x)$. Then we have
\begin{equation}\label{ps}X=\epsilon_1 U+ \epsilon_2 V\end{equation}
where $\epsilon_1,\epsilon_2$ are independent of each other and of $U,V$, and uniformly distributed on $[0,1]$. 
The probability density $p_{t+1}(x)$ will thus be found by computing the distribution of $X$ given by (\ref{ps}).
We use the following simple result
\begin{lemma} Assume $W$ is a non-negative random variable with probability density $p(x)$, and $\epsilon$ is a random variable with $\epsilon \sim Uniform ([0,1])$, $W,\epsilon$ independent. Then the probability density of the product $\epsilon W$ is given by 
\begin{equation}\label{dS}S[p](x)= \int_x^\infty \frac{p(u)}{u} du.\end{equation}
\end{lemma}
\begin{proof}
$$P(\epsilon W\leq x)= \int_0^1 \int_0^{\frac{x}{\epsilon}} p(u)du d\epsilon$$
$$= \int_0^x p(u)\int_0^{1}  d\epsilon du + \int_x^\infty p(u)\int_0^{\frac{x}{u}}  d\epsilon du$$
$$= \int_0^x p(u)du + x\int_x^\infty \frac{p(u)}{u} du$$
$$\Rightarrow\;\; \frac{d}{dx}P(\epsilon W\leq x)=\int_x^\infty \frac{p(u)}{u} du .$$
\end{proof}
Denoting the set of all probability densities on $[0,\infty)$ by ${\cal{P}}$ we can consider $S$, defined by (\ref{dS}), as an operator $S:{\cal{P}}\rightarrow {\cal{P}}$.
The above lemma implies that the density of both $\epsilon_1 U$ and $\epsilon_2 V$ is $S[p_t]$.
Therefore the density of $p_{t+1}$ of $X$, which by (\ref{ps})  is the density of the sum of two independent and identically distributed random variables $\epsilon_1 U,\epsilon_2 V$ is given by the convolution:
\begin{equation}\label{recg}p_{t+1}(x)=(S[p_t]*S[p_t])(x)=\int_0^x S[p_t](x-v)S[p_t](v)dv,\end{equation}
or more explicitly
\begin{eqnarray}
\label{explicit}p_{t+1}(x)&=&\int_0^x \Big(\int_{x-v}^\infty \frac{p_t(u)}{u} du \Big)\Big(\int_v^\infty \frac{p_t(u')}{u'} du'\Big)dv\\
&=&\int_0^x \int_y^\infty \int_{x-y}^\infty  \frac{p_t(u)}{u} \cdot  \frac{p_t(v)}{v} dv du  dy.\nonumber
\end{eqnarray}
In other words, defining the nonlinear operator $T:{\cal{P}}\rightarrow {\cal{P}} $ by
\begin{equation}\label{defT}T[p]\doteq S[p]*S[p],\end{equation}
we have that the evolution
of the wealth distribution for the Immediate Exchange model is given by
\begin{equation}\label{evol}p_{t+1}=T[p_t],\;\;\; t=0,1,2,...\end{equation}
This is the infinite-population formulation of the Immediate Exchange model.

\section{The equilibrium distribution}\label{equilibrium}

By (\ref{evol}) the equilibrium distributions are thus the solutions of the functional equation $T[p]=p$, that is
\begin{equation}\label{eqg}p=S[p]*S[p].\end{equation}
To solve this equation, we apply the Laplace transform
$${\cal{L}}[p](s)=\int_0^\infty e^{-sx}p(x)dx$$
to both sides of (\ref{eqg}), and set $\hat{p}={\cal{L}}[p]$, obtaining
$$\hat{p}(s)=\left({\cal{L}}[S[p]](s)\right)^2.$$
Noting that
\begin{equation}
\label{ls}{\cal{L}}[S[p]](s)= \frac{1}{s}\int_0^s\hat{p}(s')ds', 
\end{equation}
we conclude that the Laplace-transformed version of (\ref{eqg}) is
\begin{equation}\label{ltg}\hat{p}(s)=\Big(\frac{1}{s}\int_0^s\hat{p}(s')ds' \Big)^2.\end{equation}
To solve this equation, we set
$$g(s)=\sqrt{\hat{p}(s)}$$
and obtain that (\ref{ltg}) is equivalent to
$$g(s)=\frac{1}{s}\int_0^s(g(s'))^2ds'.$$
Multiplying both sides by $s$ and then differentiating, we obtain the differential equation
$$[sg(s)]'=g(s)^2,$$
that is
$$g'(s)=\frac{1}{s}g(s)[g(s)-1],$$
a separable equation which is solved to yield:
$$g(s)=\frac{1}{1+Cs},$$
hence
$$\hat{p}(s)=(g(s))^2=\frac{1}{(1+Cs)^2}.$$
The inverse Laplace transform now gives:
$$p(x)=\frac{1}{C^2}xe^{-\frac{x}{C}}.$$
Denoting by $w$ the mean wealth $w=\int_0^\infty xp(x)dx$, we have $C=\frac{w}{2}$, which yields
\begin{theorem}\label{tga} For each $w>0$, there exists a unique equilibrium distribution for the Immediate Exchange process satisfying $\int_0^\infty xp(x)dx=w$, given by
\begin{equation}
\label{eqd}p_w(x)=\frac{4}{w^2}xe^{-\frac{2}{w}x}.
\end{equation}
\end{theorem}
This is the Gamma distribution with shape parameter $2$, as found in the simulations of \cite{heinsalu}.

\section{Convergence to the equilibrium distribution}\label{converge}

To fully explain the simulation results in \cite{heinsalu}, we need to prove that the iterations (\ref{evol}) converge to an equilibrium distribution (\ref{eqd}), 
starting from an arbitrary initial probability density $p_0$. Since the process is wealth-preserving (see Lemma \ref{inv1} below), the value of $w$ will be determined by the mean wealth of the
initial density:
\begin{equation}\label{iw}w=\int_0^\infty xp_0(x)dx.\end{equation}
\begin{theorem}\label{convergence}
	Let $p_0(x)$ be a probability density on $[0,\infty)$ satisfying (\ref{iw}), and such that, for some $\alpha>1$, 
	\begin{equation}
	\label{fm}M_\alpha(p)=\int_0^\infty p(x)x^\alpha dx<\infty.
	\end{equation}
	Then the cumulative probability functions of the iterations (\ref{evol}) converge to that of $p_w(x)$ given by (\ref{eqd}), that is for all $x\geq 0$,
	$$\lim_{t\rightarrow \infty}\int_0^x p_t(u)du =\int_0^x p_w(u)du.$$
\end{theorem}
Since the proof of Theorem \ref{convergence} follows the same technique as that used for analogous results for the Dr\u{a}gulescu - Yakovenko model and
the Directed Random Market model \cite{katriel1,katriel2}, we will be brief, and refer to those papers for details, indicating only the 
general argument and some points where calculations somewhat different from those in the above papers are required.

For $\alpha\geq 1$ and $w>0$, we define ${\cal{P}}_{\alpha,w}$ as the set of all probability densities satisfying (\ref{iw}) and (\ref{fm}).
We first show that the operator $T$ defined by (\ref{defT})  maps the space  ${\cal{P}}_{\alpha,w}$ into itself.
\begin{lemma}\label{inv1}  If $\alpha\geq 1$, $w>0$, and $p\in {\cal{P}}_{\alpha,w}$ then $T[p]\in {\cal{P}}_{\alpha,w}$.
\end{lemma}

\begin{proof} Assume $p\in {\cal{P}}_{\alpha,w}$. Exchanging order of integration, and using the inequality  
	$(x+u)^\alpha\leq 2^{\alpha-1}(x^\alpha+u^\alpha)$, we have
	$$M_\alpha(p)=\int_0^\infty x^\alpha T[p](x)dx=\int_0^\infty x^\alpha \int_0^x S[p](x-v)S[p](v)dv dx$$ 
	$$= \int_0^\infty x^\alpha \int_0^x \Big(\int_{x-v}^\infty \frac{p(u)}{u} du \Big)\Big(\int_v^\infty \frac{p(u')}{u'} du'\Big)dv dx$$
	$$ = \int_0^\infty \Big(\int_v^\infty \frac{p(u')}{u'} du'\Big)\int_v^\infty x^\alpha  \Big(\int_{x-v}^\infty \frac{p(u)}{u} du \Big) dx dv$$
	$$ = \int_0^\infty \Big(\int_v^\infty \frac{p(u')}{u'} du'\Big)\int_0^\infty (x+v)^\alpha  \Big(\int_x^\infty \frac{p(u)}{u} du \Big) dx dv$$
	$$ \leq 2^{\alpha-1} \int_0^\infty \Big(\int_v^\infty \frac{p(u')}{u'} du'\Big)\int_0^\infty (x^\alpha+v^\alpha)  \Big(\int_x^\infty \frac{p(u)}{u} du \Big) dx dv$$
	$$ = 2^{\alpha-1} \int_0^\infty \Big(\int_v^\infty \frac{p(u')}{u'} du'\Big)\int_0^\infty  \frac{p(u)}{u}\int_0^u x^\alpha dx  du dv$$
	$$ + 2^{\alpha-1} \int_0^\infty v^\alpha\Big(\int_v^\infty \frac{p(u')}{u'} du'\Big)\int_0^\infty  \frac{p(u)}{u}\int_0^u  dx  du dv$$
	$$ = \frac{2^{\alpha-1}}{\alpha+1}M_{\alpha}(p) \int_0^\infty \int_v^\infty \frac{p(u')}{u'} du' dv+ 2^{\alpha-1} \int_0^\infty v^\alpha \int_v^\infty \frac{p(u')}{u'} du'dv$$
	$$ = \frac{2^{\alpha-1}}{\alpha+1}M_{\alpha}(p) + \frac{2^{\alpha-1}}{\alpha+1} M_{\alpha}(p)=\frac{2^{\alpha}}{\alpha+1} M_{\alpha}(p),$$
	so that $T[p]$ satisfies (\ref{fm}).
	
	Setting $\alpha=1$, the above inequality becomes an equality, and we obtain that $M_1(T[p])=M_1(p)$, so that $T[p]$ satisfies (\ref{iw}).
\end{proof}

We define the following metric on the set ${\cal{P}}_{\alpha,w}$, where we now assume $\alpha\in (1,2)$.
$$p,q\in {\cal{P}}_{\alpha,w}\;\;\;\Rightarrow\;\;\; d_{\alpha}(p,q)\doteq\sup_{s>0}\frac{|{\cal{L}}[p](s)-{\cal{L}}[q](s)|}{s^\alpha}.$$
The finiteness of $d_{\alpha,w}(p,q)$ is ensured whenever $1<\alpha<2$, see {\cite{katriel1}}, Lemma 2.3.

We use the following key estimate:

\begin{lemma}\label{lip1} If $1< \alpha<2$, $w>0$, $p,q\in {\cal{P}}_{\alpha,w}$, then
	$$d_{\alpha}(T[p],T[q])\leq \frac{2}{\alpha+1}\cdot  d_\alpha(p,q).$$
\end{lemma}

\begin{proof}
	Recalling (\ref{ls}), we have 
	$${\cal{L}}[T[p]](s)=\left({\cal{L}}[S[p]]\right)^2=\Big(\frac{1}{s}\int_0^s \hat{p}(s')ds' \Big)^2=\Big(\int_0^1 \hat{p}(s u)du \Big)^2, $$
	hence, since $|\hat{p}(s)|,|\hat{q}(s)|\leq 1$, 
	$$\frac{|{\cal{L}}[T[p]](s)-{\cal{L}}[T[q]](s)|}{s^\alpha} =\frac{1}{s^\alpha}\Big|\Big(\int_0^1 [\hat{p}(s u)- \hat{q}(su)]du \Big)\Big(\int_0^1 [\hat{p}(su)+\hat{q}(s u)]du \Big)\Big| $$
	$$\leq 2\Big(\int_0^1 u^\alpha \frac{|\hat{p}(s u)- \hat{q}(su)|}{(su)^\alpha }du \Big)\leq 2 d_{\alpha}(p,q)\int_0^1 u^\alpha du=\frac{2}{\alpha+1}\cdot d_{\alpha}(p,q),$$
	and taking the supremum over $s>0$ we obtain the result.
\end{proof}

Since $\alpha>1$ implies $\frac{2}{\alpha+1}<1$ the above Lemma implies that $T$ is contracting with respect to the metric $d_\alpha$, and, by the argument given in \cite{katriel1,katriel2}, this implies that
 that the iterates $p_t$ converge to $p_w$ in the metric $d_\alpha$, and hence in the cumulative probability sense of Theorem \ref{convergence}, concluding the proof 
 of the theorem.
 
 \section{The mixed model}\label{mixed}
 
We now discuss another model proposed in \cite{heinsalu}, which is a `mixture' of the Immediate Exchange model discussed above with a model of unidirectional
wealth transfers. 

In the unidirectional model, when two agents interact, one agent is randomly assigned to be the `loser' and the other the `winner'. The loser gives
the winner a random fraction $\epsilon$ of its wealth. Thus if $j$ is the winner then
$$x_{i}'=(1-\epsilon)x_i,\;\;\;x_j'=x_j+\epsilon x_i,$$
where $\epsilon \sim Uniform([0,1]).$
This model has recently been studied by Mart\'inez-Mart\'inez and L\'opez-Ruiz \cite{martinez} who called it the Directed Random Market. They showed that in the
infinite population limit the evolution of the wealth distribution is given by 
\begin{equation}\label{itm}p_{t+1}=T_D[p_t],\end{equation}
where
\begin{equation}\label{defTD}T_D[p](x)=\frac{1}{2}\int_0^x p_t(x-u) \int_{u}^\infty  \frac{1}{v} p_t(v) dv  du+\frac{1}{2} \int_x^\infty \frac{1}{u} p_t(u)du.\end{equation}
In \cite{katriel2} it was shown that the corresponding equilibrium distribution is the Gamma distribution with shape parameter $\frac{1}{2}$:
$$p_w(x)=\frac{1}{\sqrt{2w \pi x}}e^{-\frac{x}{2w}},$$
and convergence of the iterations (\ref{itm}) to the equilibrium distribution was proved.

The mixed model proposed in \cite{heinsalu} combines the Immediate Exchange model and the Directed Random Market model as follows: for a fixed parameter $\mu\in [0,1]$,
when two agents interact, with probability $\mu$ a unidirectional money transfer (as in the Directed Random Market model) is carried out,
and with probability $1-\mu$ a bidirectional exchange (as in the Immediate Exchange model) is carried out. 

In \cite{heinsalu} the mixed model  was investigated by simulations, and it was observed that the resulting wealth distribution is very will
fitted by a Gamma distribution with shape parameter $\alpha=2^{1-2\mu}$.  However the authors did note some deviations from the Gamma distribution.
In the extreme cases $\mu=0,\mu=1$, where the model reduces to the Immediate Exchange and to the Directed Random Market models, respectively, 
we indeed have the equilbirium distributions with shape parmeter $\alpha=2^{1-2\mu}$, as proved above and in \cite{katriel2}. However, as we will show
below, for $\mu\in (0,1)$ the equilibrium distribution is {\it{not}} a Gamma distribution.

The evolution of the wealth distribution for the mixed model will be given by $p_{t+1}=T_M[p_t]$, where
$$T_M[p]\doteq \mu T_D[p]+(1-\mu)T[p],$$
where $T$ is defined by (\ref{defT}) and $T_D$ by (\ref{defTD}). To find the equilibrium distributions we need to solve
$T_M[p]=p$, that is
\begin{equation}
\label{fpm}p=\mu T_D[p]+(1-\mu)T[p].
\end{equation}
In \cite{katriel2} it was shown that, setting $\hat{p}(s)={\cal{L}}[p](s)$, we have
$${\cal{L}}[T_D[p]](s)=\frac{1}{2s}\cdot[\hat{p}(s)+1 ]\cdot\int_0^s \hat{p}(s')ds'$$
and in Section \ref{equilibrium} we showed that
$${\cal{L}}[T[p]](s)=\left(\frac{1}{s}\cdot\int_0^s \hat{p}(s')ds'\right)^2,$$
hence applying the Laplace transform to both sides of (\ref{fpm}) gives
\begin{equation}\label{te}\hat{p}(s)=\mu \cdot\frac{1}{2s}\cdot[\hat{p}(s)+1 ]\cdot\int_0^s \hat{p}(s')ds'+(1-\mu)\left(\frac{1}{s}\cdot\int_0^s \hat{p}(s')ds'\right)^2.\end{equation}
To solve the functional equation (\ref{te}),  we define 
$$h(s)=\frac{1}{s}\cdot\int_0^s \hat{p}(s')ds', $$
so that 
\begin{equation}\label{ph}\hat{p}(s)=[sh(s)]'=sh'(s)+h(s),\end{equation}
and (\ref{te}) becomes
$$sh'(s)+h(s)=\mu \cdot\frac{1}{2}\cdot [sh'(s)+h(s)+1 ]\cdot h(s)+(1-\mu)(h(s))^2,$$
or, after rearrangement,
\begin{equation}\label{dh}h'(s)=\frac{2-\mu}{s}\cdot \frac{\left[ h(s)-1\right]h(s)}{2-\mu h(s)}.\end{equation}
This separable differential equation can be solved, but only in implicit form:
\begin{equation}\label{si}(1-h(s))^{2-\mu}=Cs^{2-\mu}\cdot (h(s))^2.\end{equation}
(\ref{si}) and (\ref{ph}) define $\hat{p}(s)$, from which the equilibrium densities $p(x)$ are obtained by Laplace inversion.
However, except in the cases $\mu=0,1$, one cannot solve (\ref{si}) for $h(s)$ in a reasonably explicit form.

To verify that the equilibrium distribution is not a Gamma distribution when $\mu\neq 0,1$, we show that the moments of the equilibrium distribution
cannot be equal to those of a Gamma distribution. The same idea was used in \cite{lallouache} with regard to a kinetic exchange model including a saving propensity. 
We compute the moments of integer order of the equilibrium distribution $p$, $M_k(p)=\int_0^\infty p(x)x^k dx$.
By (\ref{ph}) we have 
\begin{equation}\label{mk}M_k(p)=(-1)^k\hat{p}^{(k)}(0)=(-1)^k(k+1)\cdot h^{(k)}(0).\end{equation}
Thus $h(0)=1$, $h'(0)=-\frac{1}{2}M_1(p)=-\frac{w}{2}$. By successively differentiating (\ref{dh}) 
and sending $s\rightarrow 0$, we recursively compute the derivatives $h^{(k)}(0)$, $2\leq k\leq 4$. 
By (\ref{mk}) these computations give
$$M_1(p)=w,\;\;M_2(p)=\frac{3}{2-\mu}\cdot w^{2},$$
$$M_3(p)=3\cdot \frac{4+\mu}{(2-\mu)^2}\cdot w^3,\;\;M_4(p)=5\cdot\frac{\mu^2+8\mu+12}{(2-\mu)^3}\cdot w^4.$$
Denoting by $q(x)=\frac{1}{\beta^\alpha \Gamma(\alpha)}x^{\alpha-1}e^{-\frac{x}{\beta}}$ the density of a Gamma distribution, 
its moments are given by
\begin{equation}\label{mo}M_k(q)=\beta^k \prod_{j=0}^{k-1}(\alpha+j).\nonumber\end{equation}
If we wish that $M_1(q)=M_1(p)$, $M_2(q)=M_2(p)$ we need to take
$\alpha=\frac{2-\mu}{1+\mu},\;\; \beta = \frac{1+\mu}{2-\mu}\cdot w$.	
This gives
$$M_3(q)=3\cdot \frac{4+\mu}{(2-\mu)^2}\cdot w^3,\;\;M_4(q)=3\cdot\frac{2\mu^2+13\mu+20}{(2-\mu)^3}\cdot w^4.$$
While we have $M_k(p)=M_k(q)$, $1\leq k\leq 3$ (for the first two this is 
true by design, while for the third moment it is an interesting `coincidence'), the fourth moment already differs (unless $\mu=0,1$), proving that 
the equilibrium distribution is {\it{not}} a Gamma distribution.

Let us note that the Gamma distribution we fitted above by equating the first two moments gave us shape parameter $\alpha=\frac{2-\mu}{1+\mu}$, while in 
\cite{heinsalu} the fit $\alpha=2^{1-2\mu}$ was given. If one looks at these two expressions in the range $\mu\in [0,1]$, one sees they have very close values. 
Of course neither of these expressions yields the true equilibrium distribution for the mixed model, since, as shown above, this equilibrium distribution is not
Gamma.

\end{document}